\documentclass{ifacconf}

\usepackage{graphicx}      
\usepackage{natbib}        

\usepackage[cmex10]{amsmath}
\usepackage{amsfonts}

\usepackage{pgfplots}
\pgfplotsset{compat=newest}

\usepackage{tikz}

\usetikzlibrary{external}
\tikzset{external/force remake}

\newcommand{\ul}[1]{\underline{#1}}
\newcommand{\ol}[1]{\overline{#1}}

\newcommand{\R}{\mathbb{R}}
\newcommand{\C}{\mathbb{C}}

\newcommand{\cH}{\mathcal{H}}

\newcommand{\cP}{\mathcal{P}}

\newcommand{\cS}{\mathcal{S}}

\newcommand{\vx}{\mathbf{x}}
\newcommand{\vw}{\mathbf{w}}
\newcommand{\vK}{\mathbf{K}}

\newcommand{\vy}{\mathbf{y}}

\newcommand{\vz}{\mathbf{z}}

\newcommand{\myRe}{\text{Re}}

\newcommand{\vDK}{\mathbf{\Delta K}}

\newcommand{\Hinf}{\cH_{\infty}}

\DeclareMathOperator*{\argmax}{arg\,max}

\newcommand{\bigSigma}{{\overline{\sigma}}}
\newcommand{\jw}{j\omega}

\newcommand{\kit}{{(k)}}
\newcommand{\kitprev}{{(k-1)}}
\newcommand{\muit}{{(\mu)}}
\newcommand{\muitprev}{{(\mu-1)}}

\newcommand{\HinfNorm}[1]{\left\| #1 \right\|_\infty}

\newcommand{\vKmin}{\ul{\vK}}
\newcommand{\vKmax}{\ol{\vK}}

\newcommand{\myDist}{\text{dist}}

\newtheorem{theorem}{Theorem}

\newtheorem{assumption}{Assumption}
\newtheorem{corollary}{Corollary}
\newenvironment{proof}{\textit{Proof.}}{~\hfill\rule{0.75em}{0.75em}\\}

\newlength\fheight
\newlength\fwidth
\newlength\fheightTwo
\usepackage[algo2e, ruled]{algorithm2e} 

\usepackage{makecell}

\newif\ifcommenttorolf
\commenttorolftrue

\newif\ifcommenttoulrich
\commenttoulrichtrue

\begin{document}
\begin{frontmatter}

\title{Controller tuning in power systems using singular value optimization}

\thanks[authorinfo]{
	\mbox{\hspace{-0.1cm} AM (amer.mesanovic@siemens.com), RF (rolf.ﬁndeisen@ovgu.de)} are 	with the Laboratory for Systems Theory and Automatic Control, Otto-von-	Guericke-University Magdeburg, Germany, AM is also with the Siemens AG, Munich, UM (ulrich.muenz@siemens.com) is with Siemens Corp., Princeton.}

\thanks[footnoteinfo]{
	This work has been partially funded by the German Federal Ministry of Education and Research (BMBF) under Grant number 01S18066B in the frame of the AlgoRes project.}


\author{Amer Me{\v s}anovi{\'c}\quad}
\author{Ulrich M{\"u}nz \quad} 
\author{Rolf Findeisen}


\begin{abstract}                
As the share of renewable generation in large power systems continues to increase, the operation of power systems becomes increasingly challenging. The constantly shifting mix of renewable and conventional generation leads to largely changing dynamics, increasing the risk of blackouts. We propose to retune the parameters of the already present controllers in the power systems to account for the seemingly changing operating conditions. To this end, we present an approach for fast and computationally efficient tuning of parameters of structured controllers. The goal of the tuning is to shift system poles to a specified region in the complex plane, e.g. for stabilization or oscillation damping. The approach exploits singular value optimization in the frequency domain, which enables scaling to large systems and is not limited to power systems. The efficiency of the approach is shown on three systems of increasing size with multiple initial parameterizations. 
\end{abstract}

\begin{keyword}
Pole placement, singular value, optimization, stabilization, linear matrix inequalities
\end{keyword}

\end{frontmatter}

\section{Introduction}


The rising share of renewable generation in power systems leads to increased volatility and uncertainty in their operation. 
Depending on weather conditions, the power generation of renewables in the power grid varies temporally, as well as geographically. Furthermore, in the case of challenging weather conditions, conventional generation, e.g. thermal power plants, needs to compensate for the reduced renewable power generation. This leads to changing dynamics in power systems, such as time-varying oscillatory modes \citep{AlAli14, crivellaro2019beyond}. 
Doing so is challenged by the fact that existing automation systems are designed for fixed, known, dynamics \citep{kundur93a} and that control systems of components are typically parameterized manually during installation, which is a time-consuming task. 
Not adapting the system leads to time-varying dynamics in the system, which increases the risk of a blackout. Changing the existing control system entirely to increase the stability margins is typically not possible for cost reasons and as the existing systems are trusted by operators due to decades of practical operation.

We propose to address this challenge by adapting the parameters of the existing control system to the changing conditions. To this end, we propose an iterative singular value optimization based approach, which shifts poles of linear systems to desired regions.
The approach allows to tune the parameters of existing structured controllers in the system, such as PID controllers, washout filters and notch filters allowing for a nonlinear dependency on the controller parameters. Shifting poles to specific regions allows to guarantee properties such as stability, improved oscillation damping, and robustness \citep{Scherer2005}.


Notably in recent years, significant advances were made with respect to the characterization of the set of internally stabilizing controllers. Many extensions from the famous Youla parameterization \citep{youla1976modern}, such as~\cite{bamieh2002distributed} and~\cite{nayyar2013decentralized} have been achieved. Other approaches include quadratic invariance \citep{rotkowitz2005characterization}, system level approaches \citep{wang2019system}, and closed-loop response shaping \citep{zheng2019parameterization}. These approaches consider dynamic state- or output-feedback controllers with linear dynamic dependencies subject to various constraints on the controller parameters. A second area with fruitful results in recent years is $\Hinf$ optimization with stability constraints,~\cite{benner2018low,apkarian2018structured}.
Approaches for pole placement, which can shift poles to specified regions in the complex plane, typically introduce a Lyapunov matrix for the synthesis procedure \citep{Chilali96}, making the approach generally less scalable. Only few approaches consider the optimization of existing controller parameters in power systems \citep{befekadu2006robust,Marinescu09,kammer2017decentralized}. These approaches require either manual adaptation for each power system \citep{Marinescu09}, or assume specific dependencies on the controller parameters \citep{befekadu2006robust,kammer2017decentralized}.

In contrast to existing works, we propose a method for pole placement based on singular value optimization. It is based on frequency domain considerations, and thus avoids the direct formulation of (large-scale) Lyapunov matrices, making the approach generally scalable to larger systems. Previous works of the authors consider H-infinity controller synthesis, see \citep{Mesanovic17ISGT, Mesanovic18ACC, mesanovic2018optimalparameter}.

The remainder of this work is organized as follows: Section~\ref{sec.Model} derives suitable models and formulates the pole placement problem tailored towards power systems.
The proposed approach is introduced in Section~\ref{section.TuningApproach}.
In Section~\ref{Section.NumExamples} the performance of the approach is underlined considering three numerical examples with multiple initial parameterizations. Finally, conclusions are provided in Section~\ref{sec.Conclusion}.

\subsection{Mathematical preliminaries}
\label{subsec.MathPrelim}
We denote by $(\cdot)^*$  the conjugate transpose of a matrix, $\bigSigma(\cdot)$ and $\overline{\lambda} (\cdot)$ denote the largest singular value and largest eigenvalue, respectively, of a matrix. The notation $\succ$ ($\succeq$), and $\prec$ ($\preceq$) is used to denote positive (semi)definiteness and negative (semi)definiteness of  a matrix, respectively. We use $j$ to denote the imaginary unit, $\R_{\ge 0}$ denotes the set of non-negative real numbers, 
$\C$ denotes the set of complex numbers, and $\C_{>0}$ denotes the set of complex numbers with a positive real part. 
\begin{defn}[\cite{lunze2013regelungstechnik}]\label{def.MIMOPoles}
	A complex number $s_p$ is a pole of the transfer matrix $G(s) : \C \rightarrow \C^{n_y \times n_w}$, if at least one element $G_{ij}(s)$ of $G(s)$ has a pole at $s_p$.
\end{defn}
We define the distance between a point $s \in \C$ and the curve $\gamma(t) \in \C$, $t \in \R$, as $\myDist(s, \gamma) = \min_{t \in \R} | s - \gamma(t) |$.

\section{Controller tuning for power systems}
\label{sec.Model}

\begin{figure}[tb]
	\centering
	\includegraphics[width=1\columnwidth]{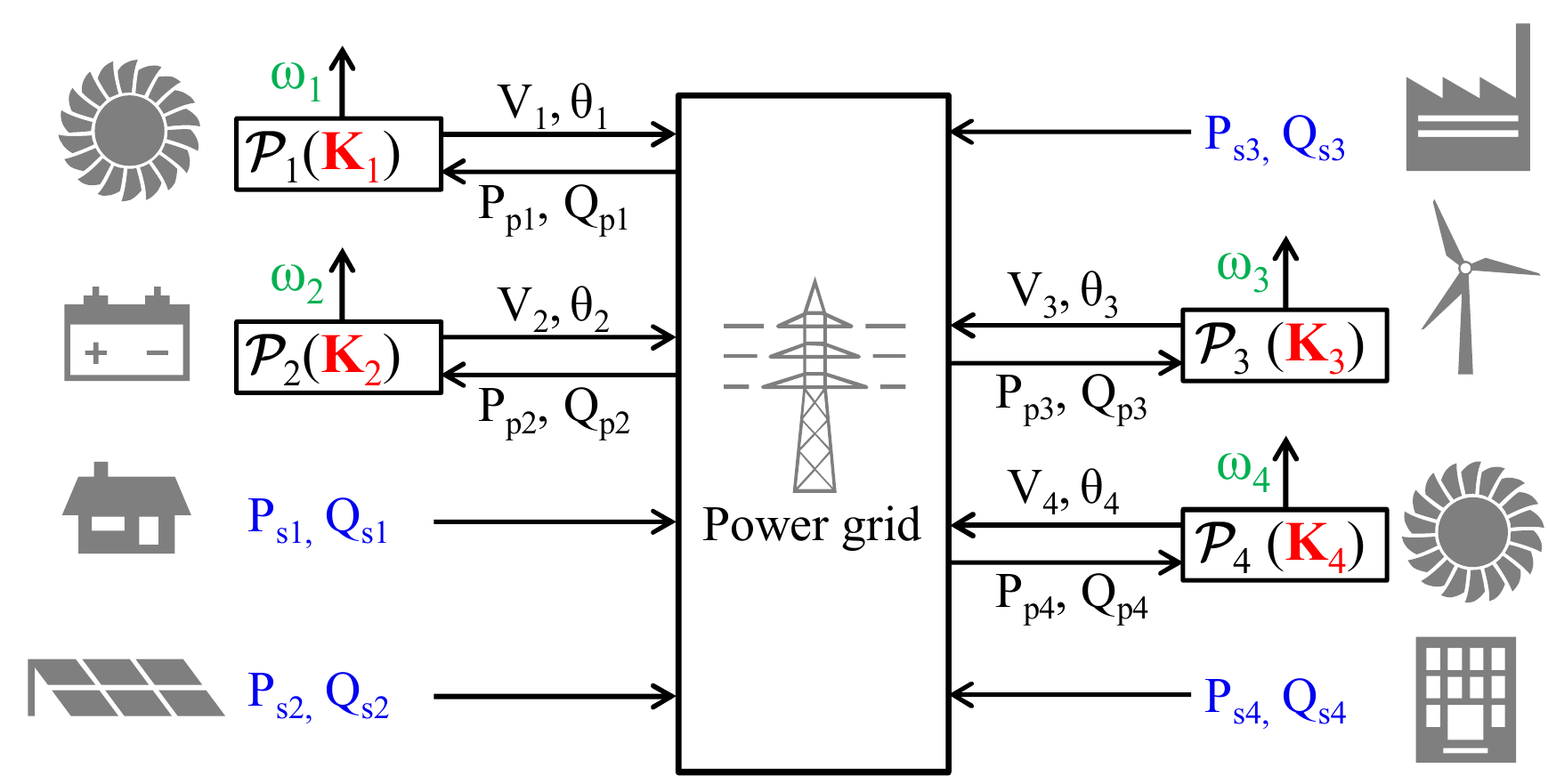}
	\caption{Exemplary power system consisting of four dynamic prosumers $\cP_i$ and four static prosumers $P_{si}$, $Q_{si}$. The tunable controller parameters $\vK_{i}$ of the dynamic prosumers are marked red. The static prosumers, marked with blue, are considered to be disturbance inputs into the system. The frequencies $\omega_i$ are outputs, marked green.}
	\label{fig.GridModel}
\end{figure}

Power systems consist of heterogeneous components, such as power plants, renewable generation, storage systems, households, and factories. For simplicity of presentation, we denote these components as prosumers, as they can either produce or consume electric power, c.f. Fig.~\ref{fig.GridModel}. We differentiate between static and dynamic prosumers, depending on whether their internal dynamic behavior is modeled. The static and dynamic prosumers are coupled through the power grid.
We outline subsequently how dynamic and static prosumers, as well as the power grid, can be modeled.

\subsection{Power system models}
\textbf{Dynamic prosumers: }
Dynamic prosumers are dynamic systems with internal states, denoted with $\cP_i$. We outline the structure of one dynamic prosumer, a power plant. Elements such as inverters or loads can be modeled as dynamics prosumers as well, however this is beyond the scope of this work. Dynamic prosumers contain tunable controller parameters $\vK_i$, marked red in Fig.~\ref{fig.GridModel}.
\begin{figure}[tb]
	\centering
	\includegraphics[width=1\columnwidth]{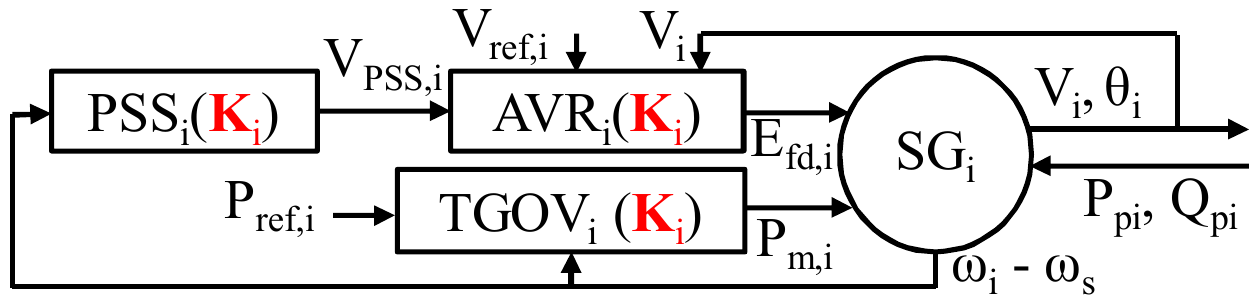}
	\caption{Simplified model of a dynamic prosumer $\cP_i$, a power plant. It consists of a synchronous generator (SG$_i$), automatic voltage regulator and exciter (AVR$_i$), a turbine and governor model (TGOV$_i$), and in some cases a power system stabilizer (PSS$_i$). Tunable parameters in the model are marked red.}
	\label{fig.SGModel}
\end{figure}
The structure of a typical power plant is shown in Fig.~\ref{fig.SGModel}. It consists of a synchronous generator (SG$_i$), an automatic voltage regulator and an exciter (AVR$_i$), a governor and turbine (TGOV$_i$), and in some cases a power system stabilizer (PSS$_i$).
AVR$_i$ controls the voltage $V_{i}$ at the power plant terminals to a reference value $V_{ref,i}$. To further improve the stability of a power system, power plants are sometimes equipped with a PSS. PSSs improve the system stability and increase the damping of oscillations by considering the plant frequency $\omega_i$ or power $P_{pi}$ in the voltage control.
TGOV$_i$ controls the generator frequency by adapting the mechanical power $P_{m,i}$ transfered to the synchronous generator. In practice, many different controllers are used \citep{IEEEExciters06}. Elements with tunable controller parameters, i.e. AVR$_i$, PSS$_i$ and TGOV$_i$, are shown as functions of $\vK_i$ in Fig.~\ref{fig.SGModel}. 


\textbf{Static prosumers: }
Static prosumers represent elements with no internal states/dynamics. They are characterized by their active power infeed $P_{si}$ and reactive power infeed $Q_{si}$, which are considered as external inputs into the model. Static prosumers, marked blue in Fig.~\ref{fig.GridModel}, often model renewable generation and loads~\citep{poolla2019placement}.
A subset of static prosumer infeeds is chosen as the disturbance input $\vw_s$ into the system.



\textbf{Power grid: } 
Dynamic and static prosumers are interconnected through the power grid, consisting of power lines, cables etc. The dynamic response of the power grid is typically orders of magnitude faster than the generation dynamics relevant for stability studies \citep{kundur93a}. For this reason, we model the grid using algebraic power flow equations
\begin{subequations} 
	\label{eq.PowerFlow}
	\begin{align}
	\!\!\!\!&P_i \! = \! \sum_{j=1}^{N_B}  V_{Bi} V_{Bj}\big( G_{cij}\cos \Delta\theta_{Bij} \! + \! B_{sij}\sin \Delta\theta_{Bij} \big)\label{eq.activepower} \\
	\!\!\!\!&Q_i \! = \! \sum_{j=1}^{N_B}  V_{Bi} V_{Bj} \big( G_{cij}\sin \Delta\theta_{Bij} \! - \! B_{sij}\cos \Delta\theta_{Bij}  \big), \label{eq.reactivepower}
	\end{align}
\end{subequations}
where $N_B$ is the number of buses (nodes) in the power system, $P_i$ and $Q_i$, are the injected active and reactive powers into the i-th bus (node) in the grid by a dynamic prosumer ($P_{pi}$, $Q_{pi}$) or a static prosumer ($P_{si}$, $Q_{si}$),
$V_{Bi}$ and $\theta_{Bi}$ are the magnitude and angle of the voltage phasor at the i-th bus from a dynamic prosumer ($V_i$, $\theta_i$) or a static prosumer ($V_{si}$, $\theta_{si}$),
and $G_{cij}$ and $B_{sij}$ are the elements of the  conductance and susceptance matrix of the grid, see \citep{kundur93a} for a detailed explanation.

\textbf{System output: }
The proposed approach requires the definition of a suitable output. In power systems, the frequencies of the dynamic prosumers, marked green in Fig.~\ref{fig.GridModel}, is defined by $\omega_i = \dot{\theta}_i$, where $\theta_i$ is the angle of the voltage phasor of $\cP_i$, are typically used to asses the system performance and stability \citep{kundur93a}. Consequently, we choose the vector of frequencies as the system output
\begin{align}
\vy = \begin{pmatrix} \omega_1 & ... & \omega_N	\end{pmatrix}^T. \label{eq.perfOutInit}
\end{align}
Here $N$ denotes the number of dynamic prosumers, and $\omega_i$ is the voltage frequency of $\cP_i$.

\subsection{Overall problem setup}
\label{subsec.CoupledModel}
When equations of prosumers are coupled via the power grid~\eqref{eq.PowerFlow}, a nonlinear differential-algebraic system of the form
\begin{subequations} \label{eq.nonlinearModel}
	\begin{align}
	\dot{\vx} =& f(\vx,\vw_s,\vK) \\
	0 =& h(\vx,\vw_s,\vK) \label{eq.nonlinearModel.alg}
	\end{align}
\end{subequations}
is obtained. Here $\vx \in \R^{\cdot N_x}$ combines all dynamic prosumer states,
$\vw_s \in \R^{n_D}$ is the vector of disturbances representing static prosumers,
$\vK \in \R^{N_K}$ is the vector of tunable controller parameters of all dynamic prosumers,
$f$ describes the prosumer dynamics, and $h$ represents the power flow equation~\eqref{eq.PowerFlow} and algebraic equations from dynamic prosumers. 

We linearize~\eqref{eq.nonlinearModel} around a known steady-state $\vx_0$ with the known input $\vw_0$. While this is an approximation, it allows us to use linear systems theory. It has furthermore been shown to be sufficient even for large-scale disturbances \citep{poolla2019placement}. After eliminating the linearized algebraic equation~\eqref{eq.nonlinearModel.alg}, 
we obtain 
\begin{subequations} \label{eq.linearizedModel}
	\begin{align}
	\dot{\vx} &=  A(\vK)  \vx +   B (\vK) \vw \\
	\vy & = C \vx.
	\end{align}
\end{subequations}
Note that the dependency of $A$ and $B$ on $\vK$ may be nonlinear.
The resulting state space system can be written in the frequency domain as
\begin{align}
G(\vK,s) = C \left(sI - A(\vK)\right)^{-1} B(\vK).
\end{align}
\begin{assumption}~\label{assum.ContParams}
	$G(\vK, s)$ is a continuous functions of the controller parameters $\vK$.
\end{assumption}
This assumption is satisfied for almost all practically relevant control elements, such as PID controllers, notch filters, lead-lag filters, washout filters etc. It does not introduce a significant restriction for the applicability of subsequent theorems.
We present in the following section an approach to tune $\vK$ in order to shift poles of $G$ to desired regions in the complex plane.

\section{Tuning approach}
\label{section.TuningApproach}

\begin{figure}[tb]
	\centering
	\includegraphics[width=0.6\columnwidth]{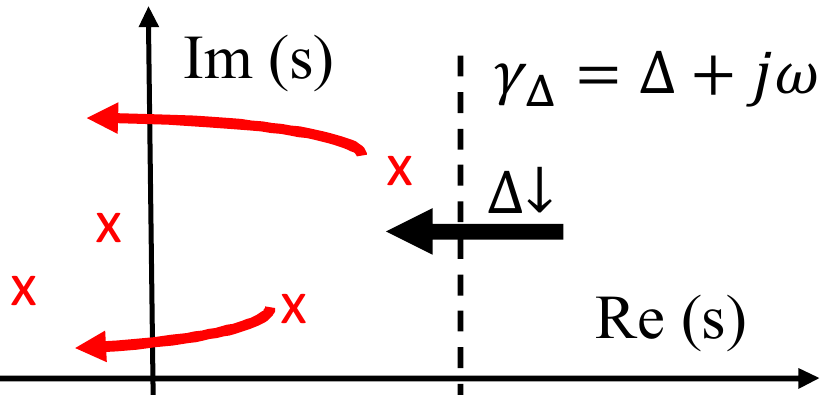}
	\caption{Visualization of the proposed approach for stabilization. By iteratively reducing $\Delta$ and minimizing $\Gamma$ along $\gamma = \Delta + \jw$, the poles, marked with red crosses, are shifted to the left half-plane.}
	\label{fig.AlgStabVis}
\end{figure}

For the proposed approach, a continuous ``optimization curve'' $\gamma (t)$, $t \in \R$ must first be defined, which is used to place the poles into a desired region. One example for $\gamma$ is a vertical axis in the complex plane $\gamma_\Delta(\omega)= \Delta + \jw$.
We denote the maximum of $\bigSigma(G(s))$ along $\gamma (t)$ as 
\begin{align}
\Gamma(\vK) = \max_t \:\: \bigSigma(G(\gamma(t))).
\end{align}
In the proposed approach, we show that by minimizing $\Gamma(\vK)$, the distance between the optimization curve $\gamma(t)$ and sufficiently close poles of $G$ is increased. Thus, by iteratively shifting $\gamma$, poles of $G$ can be placed into desired regions in the complex plane. Figure~\ref{fig.AlgStabVis} illustrates the basic idea of the approach with the left half-plane as the desired region. For this purpose, the curve $\gamma_\Delta(\omega)= \Delta + \jw$ is used and iteratively shifted to left.

We introduce in this section several results which enable the formulation of the pole placement approach:
\begin{itemize}
	\item We first show in Theorem~\ref{thm.SigmaApprox} how $\bigSigma(G)$ can be approximated in the neighborhood of one pole.
	\item We then show in Theorem~\ref{lemma.PolesNoCross} that minimization of $\Gamma$ will not allow poles to cross to the other side of $\gamma$.
	\item Theorem~\ref{thm.PoleShift} shows how the minimization of $\Gamma$ will increase the distance between a sufficiently close pole and $\gamma$.
	\item Finally, Theorem~\ref{thm.GammaMin} introduces an approach to minimize $\Gamma$.
\end{itemize} 
To shorten the notation, we avoid explicitly writing the dependency of $G$ on $\vK$. It is assumed, however that $G$ is always a function of $\vK$.

\begin{theorem}\label{thm.SigmaApprox}
	Given a detectable multi-input multi-output system $G(s)$ with $n$ rows and $m$ columns, and with a set of poles $\cS_p$. Then, $\bigSigma(G(s))$ can be approximated in a sufficiently small neighborhood of one pole $s_{p} \in \cS_p$ by
	\begin{align}
		\bigSigma (G(s)) \approx \frac{a}{|s-s_p|^{n_p}}, \label{eq.SigmaApprox}
	\end{align}
	where $n_p$ is the largest multiplicity of $s_p$ in any single-input single-output transfer function $G_{ij}(s)$ in $G(s)$, and $a \in \R_{\ge 0}$.
\end{theorem}
\begin{proof}
	Per definition \citep{hopcroft2012computer}
	\begin{align}
		\bigSigma(G(s)) \! = \! \max_{\|\vz\|_2 = 1} \! \| G(s) \vz \|_2 
		 \! = \! \! \max_{\|\vz\|_2 = 1} \sqrt{\sum_{i = 1}^{n} \left( \sum_{j = 1}^{m} G_{ij}(s) z_j \right)^2 }. \label{eq.BigSigmaDef}
	\end{align}
	Without restriction of generality, we express every transfer function $G_{ij}(s)$ as
	\begin{align}
		G_{ij}(s) = \frac{N_{ij}(s)}{(s - s_p)^{n_{ij}} D_{ij}(s)}.
	\end{align}
	Here $n_{ij}$ is the multiplicity of $s_p$ in $G_{ij}$, where $n_{ij}$ may be zero, and $N_{ij}(s)$, $D_{ij}(s)$ are polynomials with no zeros in a neighborhood of $s_p$. In a sufficiently small neighborhood of $s_p$, $G_{ij}$ can be approximated as
	\begin{align}
		G_{ij}(s) \approx \frac{N_{ij}(s_p)}{(s - s_p)^{n_{ij}} D_{ij}(s_p)}  = \frac{a_{ij}}{(s - s_p)^{n_{ij}}}.
	\end{align}
	From this approximation and~\eqref{eq.BigSigmaDef} it follows that
	\begin{align}
	\bigSigma(G(s)) \approx &  \frac{1}{|s - s_p|^{n_p}} \times \\   
	&\max_{\|\vz\|_2 = 1} \sqrt{\sum_{i = 1}^{n} \left( \sum_{j = 1}^{m} a_{ij} (s - s_p)^{n_p - n_{ij}} z_j \right)^2 }. \nonumber
	\end{align}
	In a sufficiently small neighborhood of $s_p$, $(s - s_p)^{n_p - n_{ij}}$ is approximately zero if $n_{ij}<n_p$, leading to
	\begin{align}
		\bigSigma(G(s)) \approx \frac{1}{|s - s_p|^{n_p}} \max_{\|\vz\|_2 = 1} \sqrt{\sum \left( \sum a_{kl} z_j \right)^2 }.
	\end{align}
	Here $a_{kl}$, with a small abuse of notation, denotes a transfer function in which $s_p$ has a multiplicity $n_p$.
	 By denoting the solution of the maximization problem in the previous equation as $a \in \R_{\ge 0}$,we obtain~\eqref{eq.SigmaApprox}. 
\end{proof}


%

\begin{corollary} \label{corol.GammaLarge}
	Given a detectable system $G(s)$, and the set of poles $\cS_p$ of $G$. 
	If a system pole $s_{p} \in \cS_P$ is sufficiently close to $\gamma$, i.e. if $\myDist(s_p, \gamma)$ is sufficiently small, then $\Gamma$ is reached in a neighborhood of $s_p$.
\end{corollary}
\begin{proof}
	As $\bigSigma(G(s))$ is a continuous function of $s$ \citep{de1989analytic}, and $\bigSigma(G(s))$ reaches $+\infty$ at $s_p$ (Theorem~\ref{thm.SigmaApprox}), the claim can be trivially shown. The full proof is left out due to page restrictions.
\end{proof}

To clarify the claim of Corollary~\ref{corol.GammaLarge}, we consider the system
\begin{align}
G'(s) = \begin{pmatrix}
\frac{s}{(s+1)(s+2)} & \frac{s-3}{s^2 + 3s + 3} \\
\frac{s^2+4s + 10}{(s^2 - s + 1)} & \frac{s+4}{(s+1)(2s-1)}.	
\end{pmatrix}
\end{align}
This system has the pole set $\cS_p' = \{-1, -2, 0.5, -1.5 \pm j 0.87, 0.5 \pm j 0.87\}$.
Figure~\ref{fig.examplePoles} shows the largest singular value plot of $G'(s )$.
\begin{figure}[tb]
	\centering
	\includegraphics[width=1\columnwidth]{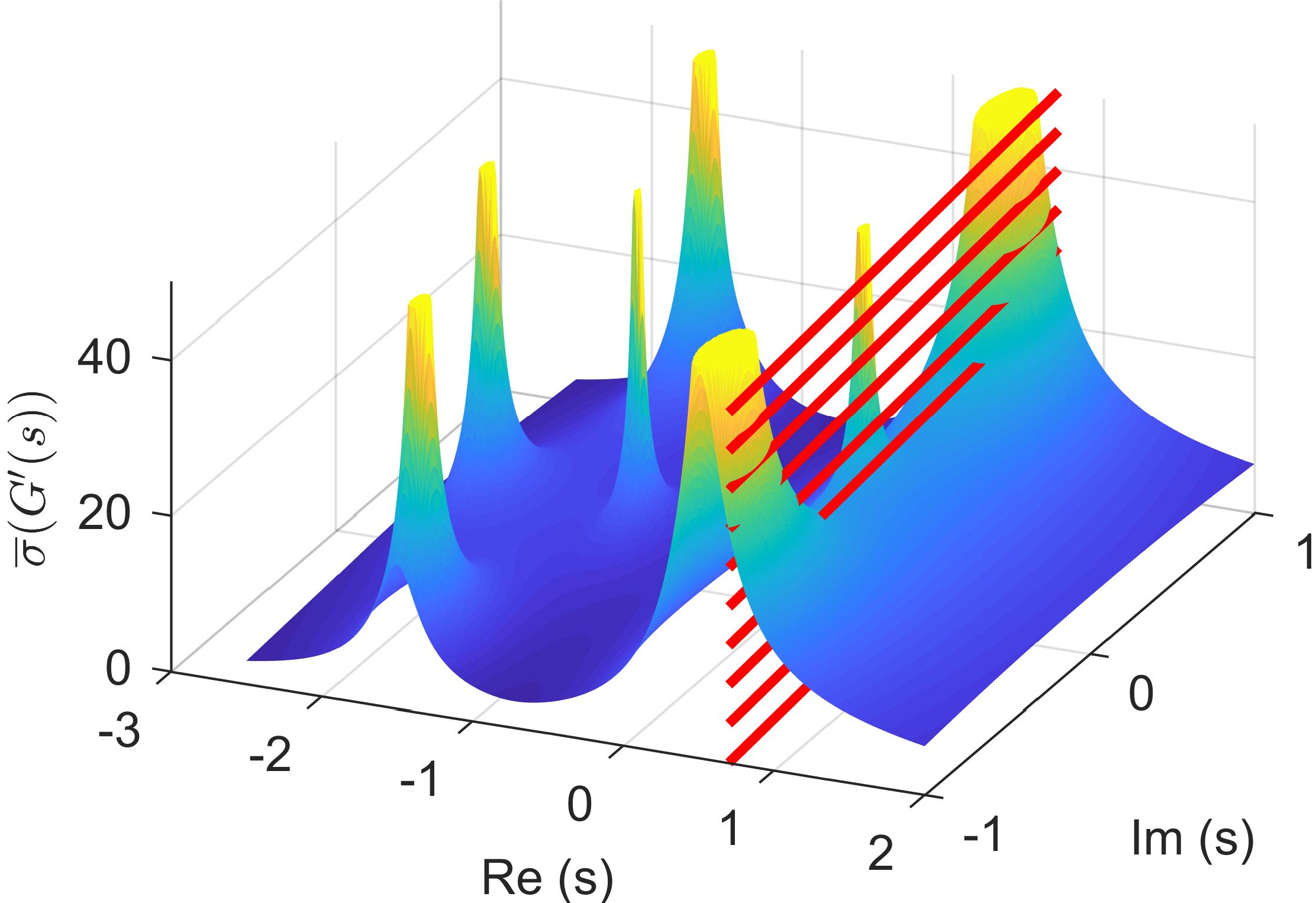}
	\caption{Graphical representation of $\bigSigma(G'(s))$; $\bigSigma(G'(s))$ approaches infinity in the surrounding of any $s_{pij}\in \tilde \cS$. Red lines visualize the optimization curve $\gamma(\omega) = 0.7 + \jw$.}
	\label{fig.examplePoles}
\end{figure}
Red lines visualize the optimization curve $\gamma(\omega) = 0.7 + \jw$, along which $\bigSigma(G(s))$ is evaluated. It shows that $\bigSigma(G(\gamma(t))$ approaches high values in the proximity of system poles (Corollary~\ref{corol.GammaLarge}).
\begin{theorem} \label{lemma.PolesNoCross}
	Given a detectable system $G(\vK, s)$ with the set of poles $\cS_p(\vK)$ that satisfies Assumption~\ref{assum.ContParams}. Furthermore, assuming that no parameter-dependent pole-zero cancellations occur on $\gamma$. Then, $s_{p} \in \cS_p$ will not cross $\gamma$ during the minimization of $\Gamma$.
\end{theorem}
\begin{proof}
	The proof is analogous to the stability-certificate provided in \citep{mesanovic2018optimalparameter} and is not shown in detail due to page limitations. It is based on the fact that, due to Assumption~\ref{assum.ContParams}, $s_{pij}$ is a continuous function of $\vK$, and cannot discretely ``jump over'' $\gamma(t)$. Thus, for $s_{pij}$ to cross $\gamma$, $\Gamma$ would have to be increased to very high values (Corollary~\ref{corol.GammaLarge}) in at least one step during the minimization, which is not allowed by the optimization method.
\end{proof}
The assumption about pole-zero cancellations does not restrict generality, because if this occurs, $\gamma$ can be shifted to avoid the cancellation and the minimization can be repeated.
\begin{theorem}[pole shifting] \label{thm.PoleShift}
	Given a detectable system\\$G(\vK, s)$, with the set of poles $\cS_p(\vK)$, and an optimization curve $\gamma(t)$. Furthermore, given two parameterizations $\vK_1$ and $\vK_2$ and any pole $s_p (\vK_1) \in \cS_p(\vK_1)$ such that $\myDist(s_p(\vK_1), \gamma)$ is sufficiently small and that $\Gamma(\vK_1)$ and $\Gamma(\vK_2)$ are reached in the neighborhood of $s_p$.
	Under Assumption~\ref{assum.ContParams}, 
	if $\Gamma(\vK_2)$ is sufficiently smaller than $\Gamma(\vK_1)$,
	the distance of $s_p(\vK_2)$ to $\gamma$ is greater than the distance of $s_p(\vK_1)$ to $\gamma$, i.e.
	\begin{align} \label{eq.PoleShift}
	\myDist(s_p(\vK_1), \gamma) < \myDist(s_p(\vK_2), \gamma).
	\end{align}
\end{theorem}
\begin{proof}
	If $\gamma$ is sufficiently close to $s_p$, then $\Gamma$ is achieved in the neighborhood of $s_p$ (Corollary~\ref{corol.GammaLarge}), and equals to (Theorem~\ref{thm.SigmaApprox})
	\begin{align}
		\Gamma(\vK) \approx \max_t \frac{a(\vK)}{|s_p - \gamma(t)|^{n_p}}  = \frac{a(\vK)}{\myDist(s_p(\vK),\gamma)^{n_p}}. \label{eq.approxGamma}
	\end{align}
	From $\Gamma(\vK_1) > \Gamma(\vK_2)$, we obtain
	\begin{align}
		\frac{a(\vK_1)\eta}{\myDist(s_p(\vK_1),\gamma)^{n_p}} = \frac{a(\vK_2) \cdot }{\myDist(s_p(\vK_2),\gamma)^{n_p}},
	\end{align}
	where $\eta<1$ quantifies the reduction of $\Gamma$. From the previous equation, it follows
	\begin{align}
	\frac{\myDist(s_p(\vK_2),\gamma)}{\myDist(s_p(\vK_1),\gamma)} = \left(\frac{a(\vK_2) }{a(\vK_1)\cdot\eta} \right)^{1/n_p}. \label{eq.OptImprovementPush}
	\end{align}
	Thus, if $\eta$ is sufficiently small such that $\frac{a(\vK_2) }{a(\vK_1)\cdot\eta} >1$,~\eqref{eq.PoleShift} directly follows.
\end{proof}

Note that a sufficient value for $\eta$ depends on the parameter dependency of $a(\vK)$: ideally, if changes of $a(\vK)$ remain small,~\eqref{eq.OptImprovementPush} is equivalent to $\eta<1$. From~\eqref{eq.approxGamma}, it follows that $\partial \: \Gamma / \partial \: \myDist(s_p(\vK),\gamma) \gg \partial \: \Gamma / \partial \: a$ in a neighborhood of $s_p$. Thus, $\gamma$ can be chosen sufficiently close to $s_p$, such that the steepest descent direction during optimization leads to the displacement of $s_p(\vK)$, instead of the minimization of $a(\vK)$.

The following consequence arises from Theorem~\ref{thm.PoleShift}:
\begin{corollary}
	The minimization of the $\Hinf$ norm of a stable system $G(s)$, i.e. $\HinfNorm{G(\vK,s)}$, improves damping of poles which are sufficiently close to the imaginary axis.
\end{corollary}
\begin{proof}
	The proof is obtained by inserting $\gamma(\omega) = \jw$ into Theorem~\ref{thm.PoleShift}. For this case, $\Gamma(\vK) = \HinfNorm{G(\vK,s)} = \max_{\omega \in \R} \bigSigma(G(\jw)) $, and its minimization will push the poles away from the imaginary axis.
\end{proof}

The previous theorem shows that minimization of $\Gamma$ will displace sufficiently close poles away from $\gamma$. However, an appropriate method to minimize $\Gamma$ is needed. 
We propose to do so by applying a similar optimization procedure as presented in~\citep{mesanovic2018optimalparameter}.
\begin{theorem}[$\Gamma$ minimization] \label{thm.GammaMin}
	Given a detectable multiple-input-multiple-output system with the transfer function $G(\vK, s)$ which satisfies Assumption~\ref{assum.ContParams}. Furthermore, given a continuous function $\gamma(t) \in \C$, $t \in \R$. Assuming there are no cancellations of parameter-dependent poles and zeros on $\gamma$, 
	the solution of
	\begin{subequations} \label{eq.GammaFiniteOptProb}
		\begin{align}
		\!\!\!\!\!\!\!\min_{\Gamma, \vK}  & \quad \Gamma \\
		\!\!\!\!\!\!\!\text{s.t.} \quad & \begin{pmatrix}
		\Gamma I     & G(\vK, \gamma(t_k)) \\
		G(\vK, \gamma(t_k))^* & \Gamma I
		\end{pmatrix} \succ 0, \quad \forall t_k \in \Omega \label{eq.GammafiniteConstraints} \\
		& \ul{\vK} \leq \vK \leq \ol{\vK}
		\end{align}
	\end{subequations}
	minimizes $\Gamma (\vK) = \max_t \:  \bigSigma(G(\vK,\gamma(t)))$ for the sampling set $\Omega$. Here $\ul{\vK}$ and $\ol{\vK}$ are box constraints on the controller parameters, which may be $\pm \infty$.
\end{theorem}
\begin{proof}
	Since $\bigSigma \: (G(\gamma(t)))^2 = \overline{\lambda} (G(\gamma(t))^* G(\gamma(t)))$, it follows
	\begin{align}
	& \max_t \: \bigSigma(G(\vK,\gamma(t))) \le \Gamma \\
	\Leftrightarrow \quad & \overline{\lambda} (G(\gamma(t))^* G(\gamma(t))) < \gamma^2, \quad \forall t \in \R\\
	\Leftrightarrow \quad &  G(\gamma(t))^* G(\gamma(t)) - \gamma^2 I \prec 0, \quad \forall t \in \R.
	\end{align}
	By using the Schur complement on the last expression and by discretization, we obtain~\eqref{eq.GammafiniteConstraints}, which is the basis for~\eqref{eq.GammaFiniteOptProb}.
\end{proof}
Note that $\Omega$ does not have to span the entire $\R$. It only needs to be defined around poles which are close to $\gamma(t)$, where $\Gamma$ is reached. Furthermore, $\Omega$ does not have to be infinitely dense. Due to Corollary~\ref{corol.GammaLarge}, it suffices that $\exists t_k \in \Omega$, such that $\gamma(t_k)$ is in a sufficiently small neighborhood of $s_p$. This way, the rise of $\Gamma$ due to $s_p$ is captured by the discretization. Thus, only several values need to be in $\Omega$ around poles close to $\gamma$, which allows for efficient optimization.

However, Problem~\eqref{eq.GammaFiniteOptProb} is still non-convex due to the nonlinear parameter dependency in $G(\vK,s)$. To use linear matrix inequality solvers, we transform~\eqref{eq.GammaFiniteOptProb} into a series of convex optimization problems. In each iteration, $G$ is linearized around the parameter vector $\vK^\kitprev$ from the previous iteration, and $G_L^\kit(\vK,s)$ is obtained. The following problem is then iteratively solved
\begin{subequations} \label{eq.GammaFiniteOptProbConvex}
	\begin{align}
	& \min_{\Gamma, \vK}   \quad \Gamma \\
	& \text{s.t.} \quad  \begin{pmatrix}
	\Gamma I     & G_L^\kit(\vK, \gamma(t_k)) \\
	G_L^\kit(\vK, \gamma(t_k))^* & \Gamma I
	\end{pmatrix} \succ 0, \quad \forall t_k \in \Omega\\
	& \qquad \ul{\vK} \leq \vK \leq \ol{\vK} \qquad |\vK - \vK^\kitprev| \leq \vDK.
	\end{align}
\end{subequations}
Here the last constraint is added to preserve the linearization accuracy.
Algorithm~\ref{algo.GammaMinimization} summarizes the proposed approach to minimize $\Gamma$. 
\begin{algorithm2e}[h]\label{algo.GammaMinimization}
	\KwIn{$G$, $\vK_{0}$, $\vDK$, $k_{max}$, $\gamma$} 
	$k=1$, choose $0< \alpha < 1$\;
	\While{$k\le k_{max}$ or not converged}{
	$G_{L}^\kit (\vK,s) = $  linearization of the parametric dependency of $G(\vK,s)$ around $\vK^\kitprev$\;
	  	$\vK^\kit = $ solution of~\eqref{eq.GammaFiniteOptProbConvex}\;
	  	\If{$\Gamma(\vK^\kit) \ge \Gamma(\vK^\kitprev)$ or a pole crossed $\gamma$}{
	  		$\vDK = \vDK \times \alpha$\label{st.TRAdapt} - increase linearization accuracy\;
	  		$\vK^\kit = \vK^\kitprev$\;
		}
	$k = k+1$
	}
	\Return{$\vK^\kitprev$}
	\caption{Minimization of $\Gamma(\vK)$.}
\end{algorithm2e}

For the proposed iterative pole placement approach, a suitable $\gamma$ is chosen in each iteration, such that minimization of $\Gamma$ pushes the poles towards the target region.  

In the following, we focus on the question of stabilization, for which the target region for all poles is the left half-plane, as this case is considered in subsequent numerical studies. To this end, a family of vertical lines, described with $\gamma_\Delta(\omega) = \Delta + \jw$, is considered. To stabilize the system, $\Delta$ is gradually reduced until all poles are in the left half-plane, as illustrated in Fig.~\ref{fig.AlgStabVis}.

Algorithm~\ref{algo.SVStabilization} outlines the proposed algorithm to achieve stability.
\begin{algorithm2e}[h]\label{algo.SVStabilization}
	\SetAlgoLined
	\KwResult{Stabilizing controller parameterization}
	\KwIn{$G(\vK)$, $\vK^{(0)}$, $\vKmin$, $\vKmax$}
	$\mu = 1$;  $s_{p,max}^{(0)} = \argmax_{s_p \in \cS_p(\vK^{(0)})} \: \myRe(s_p)$\;
	\While{$\myRe(s_{p,max}^\muitprev)>0$}{
		Select $\Delta^\muit > \myRe(s_{p,max}^\muitprev)$ s.t. $\gamma_{\Delta}^\muit(\omega) = \Delta^\muit + \jw$ is sufficiently close to $s_{p,max}^\muitprev$\;
		$\vK^\muit = $ solution of Algorithm~\ref{algo.GammaMinimization} for $\gamma^\muit(\omega)$ 
		$s_{p,max}^\muit = \argmax_{s_p \in \cS_p(\vK^\muit)} \: \myRe(s_p)$\;
		$\mu = \mu+1$\;
	}
	\Return{$\vK^\muitprev$}\;
	\caption{Stabilization}
\end{algorithm2e}
In the $\mu$-th iteration, the pole with the maximal real value $s_{p,max}^\muitprev$ is calculated. Then, $\Delta^\muit$ is chosen such that the maximum of $\Gamma^\muit$ is reached in the neighborhood of $s_{p,max}^\muitprev$. Thereby, $\Delta^\muit$ should be greater than $s_{p,max}^\muitprev$ to push the poles to the left half-plane. The algorithm terminates when all poles are in the left half-plane.
Determining $\Delta^\muit$ is important for the efficiency of the algorithm. If it is too small, $\Gamma^\muit$ must be minimized by a large amount to push the pole away, due to the high local gradients around poles. On the other hand, if it is too far from $s_{p,max}^\muitprev$, $\Gamma^\muit$ will not be reached in the neighborhood of $s_{p,max}^\muitprev$, and minimization of $\Gamma$ is not guaranteed to push the pole away from $\gamma$.  We determine $\Delta$ by first finding the largest value $\Delta_{max}$ via a line-search, such that the maximum of $\gamma$ is still achieved in the neighborhood of $s_{p,max}^\muitprev$. Then, we set $\Delta^\muit = 0.5\left(\Delta_{max} + \myRe(s_{p,max}^\muitprev)\right)$, which solves both of the previous issues in the considered numerical examples.

\section{Application examples}
\label{Section.NumExamples}
We show the applicability of the approach on three systems with different initial parameterizations\footnote{For the optimization, we use a Windows computer with an Intel$^\circledR$ i7-4810MQ CPU running at 2.8 GHz and with 8 GB of RAM.}.
Note that the presented computation times should only give a reference about the computational complexity, as more tailored methods and special hardware would allow to further decrease the computation time. All controllers in the considered systems are purely local with different structures.

\subsection{Four power plant model}

We first consider a four power plant system, c.f. \citep{kundur93a}, as shown in Fig.~\ref{fig.KundurGrid}. Detailed parameters of the system and the initial controller parameterization $\vK_{i,4}$ can be found in \citep{Mesanovic17ISGT}. The system consists of 56 states and 32 controller parameters. Static prosumers, representing disturbance inputs, are marked blue in Fig.~\ref{fig.KundurGrid}.

For comparison, we use an approach which is based on the use of a Lyapunov matrix $P$ in the optimization problem, where the condition $A(\vK)^T P + PA(\vK)\prec 0, P\succ 0$ is used for stability. It is solved using an iterative coordinate-descent method, often called PK iteration, with iterative parameter linearization. 
Subsequently, we denote this as the PK method.
\begin{figure}[tb]
	\centering
	\includegraphics[width=1\columnwidth]{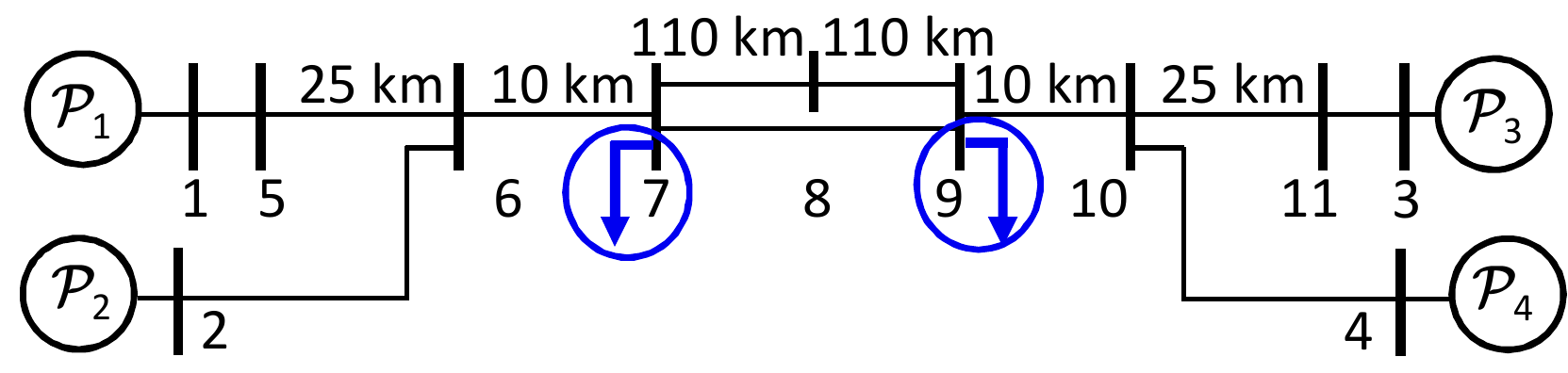}
	\caption{Power system with four power plants from \citep{kundur93a}. Static prosumers, which are modeled as disturbance inputs, are marked blue.}
	\label{fig.KundurGrid}
\end{figure}
\begin{table}[tb]
	\begin{center}
		\caption{Optimization results for the \mbox{4-prosumer} system with the SV method.}\label{table.KundurSV}
		\begin{tabular}{ccccc}
			\thead{Init.\\ param.} & \thead{\# unst.\\ poles} & \textbf{$\mathbf{\max \myRe(s_{pij})}$} & \thead{Comp.\\time} & \thead{\# iter.} \\ \hline
			 $1.25\:\vK_{i,4}$   &   2    &    0.058      & 1.8 s    &     1    \\
			 $1.5\:\vK_{i,4}$   &   2    &    0.43     & 5 s    &     1    \\ 
			 $2\:\vK_{i,4}$   &   5    &    1.3     & 6.4 s    &     2    \\\hline
		\end{tabular}
	\end{center}
\end{table}
\begin{table}[b]
	\begin{center}
		\caption{Optimization results for the \mbox{4-prosumer} system with the PK method.}\label{table.KundurPK}
		\begin{tabular}{ccccc}
			\thead{Init.\\ param.} & \thead{\# unst.\\ poles} & \textbf{$\mathbf{\max \myRe(s_{pij})}$} & \thead{Comp.\\time} & \thead{\# iter.} \\ \hline
			$1.25\:\vK_{i,4}$   &   2    &    0.058      & 180 s    &     6    \\
			$1.5\:\vK_{i,4}$   &   2    &    0.43     & 215 s    &     7    \\ \hline
		\end{tabular}
	\end{center}
\end{table}

Table~\ref{table.KundurSV} shows the results of the optimization with the proposed singular value (SV) method for different unstable initial parameterizations, obtained by multiplying $\vK_{i,4}$ with a scaling factor. The computation time depends on $\max \myRe(s_{pij})$, as more iterations are necessary if the unstable poles have a larger real part. For the three tested parameterizations, the SV method stabilized the system in 1.8 s, 5 s, and 6.4 s, respectively. On the other hand, the PK method, whose results are shown in Table~\ref{table.KundurPK}, required 180 s, and 215 s, respectively, while it could not stabilize the system in the third case. Improvements per iteration for the SV and PK method are shown in Fig~\ref{fig.KundurPoles}. Thus, the SV method significantly outperforms the PK method for this system.

\begin{figure}[tb]
	\centering
	\includegraphics[width=1\columnwidth]{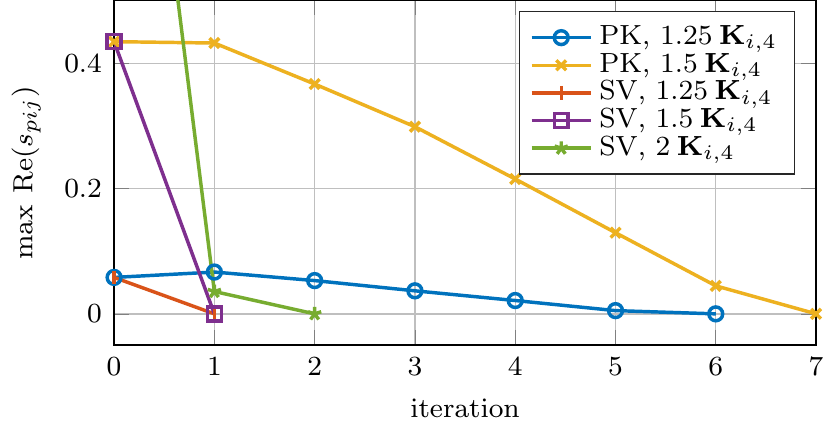}
	\caption{Improvement of $\max \: \myRe(s_{pij})$ per iteration with different initial parameterizations for the 4-prosumer system.} 
	\label{fig.KundurPoles}
\end{figure}

\subsection{The IEEE 39 bus 10 power plant model}

We next consider the IEEE 39-bus system, consisting of 10 power plants, c.f. Fig~\ref{fig.IEEE39}. The parameters of the system can be found in \citep{Mesanovic17ISGT}. It consists of 190 states and 100 controller parameters. We denote the initial controller parameterization from \citep{Mesanovic17ISGT} as $\vK_{i,10}$.
\begin{figure}[tb]
	\centering
	\includegraphics[width=1\columnwidth]{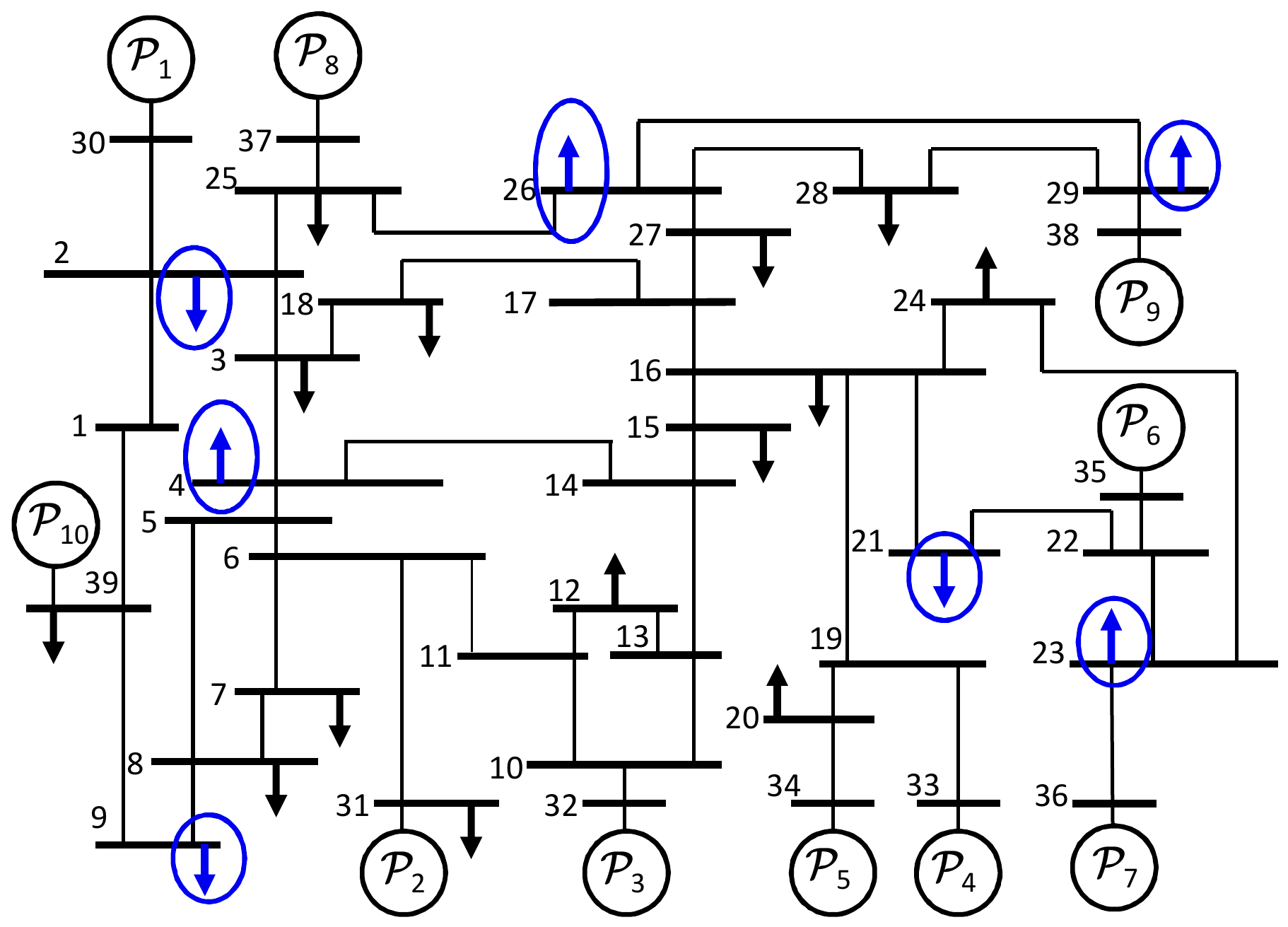}
	\caption{IEEE 39-bus system with 10 dynamic prosumers. Static prosumers modeled as disturbance inputs are marked blue.}
	\label{fig.IEEE39}
\end{figure}
\begin{table}[tb]
	\begin{center}
		\caption{Optimization results for the \mbox{10-prosumer} system with the SV method.}\label{table.IEEESV}
		\begin{tabular}{ccccc}
			\thead{Init.\\ param.} & \thead{\# unst.\\ poles} & \textbf{$\mathbf{\max \myRe(s_{pij})}$} & \thead{Comp.\\time} & \thead{\# iter.} \\ \hline
			$1.5\:\vK_{i,10}$   &   7    &    0.54      & 19 s    &     5    \\
			$1.7\:\vK_{i,10}$   &   11    &    1.83     & 24 s    &     4    \\ 
			$2\:\vK_{i,10}$   &   12    &    4.7     & 65 s    &     12    \\
			$3\:\vK_{i,10}$   &   18    &    17.9     & 259 s    &     12    \\\hline
		\end{tabular}
	\end{center}
\end{table}

Table~\ref{table.IEEESV} shows the computation results with the SV method. It shows that, even when the poles are far in the right half-plane, a stabilizing parameterization can be found in less than 5 minutes. Even though the approach requires the same amount of iterations to stabilize the system with $2\:\vK_{i,10}$ and $3\:\vK_{i,10}$, the computation times are different because the internal minimization of $\Gamma$ requires more time per iteration for $3\:\vK_{i,10}$.
For a system of this size, the PK method did not provide any results within reasonable time. Figure~\ref{fig.IEEE39Poles} shows the improvement in each iteration for the considered system.

\begin{figure}[tb]
	\centering
	\includegraphics[width=1\columnwidth]{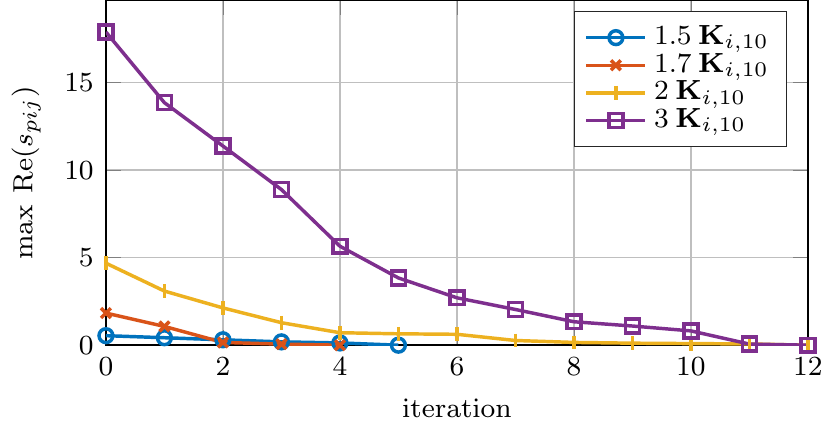}
	\caption{Improvement of $\max \: \myRe(s_{pij})$ per iteration with different initial parameterizations using the SV method for the 10-prosumer system.}
	\label{fig.IEEE39Poles}
\end{figure}

\subsection{European 53 power plant model}

Finally, we consider a model with 53 power plants and 35 buses, developed within the research project~\cite{dynagrid}.
The controllers used for this model are presented in \citep{mesanovic2018optimalparameter}.
A more detailed description of the considered system cannot be presented due to space limitations. The system has a total of 469 states, 116 controller parameters, and 15 disturbance inputs.

\begin{table}[tb]
	\begin{center}
		\caption{Optimization results for the \mbox{53-prosumer} system with the SV method.}\label{table.DynaGridSV}
		\begin{tabular}{ccccc}
			\thead{Init.\\ param.} & \thead{\# unst.\\ poles} & \textbf{$\mathbf{\max \myRe(s_{pij})}$} & \thead{Comp.\\time} & \thead{\# iter.} \\ \hline
			$\vK_{1,53}$   &   4    &    0.26      & 250 s    &     1    \\
			$\vK_{2,53}$   &   4    &    0.5      & 470 s    &     2    \\
		\end{tabular}
	\end{center}
\end{table}

Table~\ref{table.DynaGridSV} shows the computation results with the SV method for two initial parameterizations. Note that the computation times are significantly impacted by the calculation of $\Delta^\muit$ in each iteration, requiring from 50 s to 100 s per iteration. Still, the SV method can find a stabilizing solution for this large system in a reasonable time frame.

\section{Conclusion}
\label{sec.Conclusion}

As the share of renewable generation in power system rises, new methods for power system control become necessary to reduce the risk for a blackout. We proposed an approach for pole placement  formulated in the frequency domain. As it does not introduce a Lyapunov matrix, the approach generally scales nicely. Nonlinear parameter dependencies can be considered as well. The approach is based on iterative minimization of the largest singular value along an optimization curve. The applicability of the approach was demonstrated in three simulation studies with different initial parameterizations. Future work will consider more efficient methods to determine $\Delta$, as well as validation for other target regions.


%


\bibliography{references_amer_mesanovic}{}   
                                       
\end{document}